\RequirePackage{fix-cm}
\documentclass[natbib]{svjour3}                     % onecolumn (standard format)
\smartqed  % flush right qed marks, e.g. at end of proof
\usepackage[american]{babel}		%English hyphenation etc.
\usepackage[utf8]{inputenc}		%Input encoding
\usepackage[T1]{fontenc}		%Font encoding
\usepackage[autostyle,english=american]{csquotes}	%Quotation marks ' '
\usepackage{graphicx}			%Include graphics
\usepackage{subfig}				%Subfigures
\usepackage{marvosym}			%Special symbols

\usepackage{amsmath}
\usepackage{amsfonts}
\usepackage{amssymb}
\usepackage{mathtools}
\usepackage{nicefrac}
\usepackage{algorithm2e}
\usepackage{booktabs}

\usepackage{pslatex}
\usepackage{xcolor}
\begin{document}

\title{On the Shapley Value of Unrooted Phylogenetic Trees}



%\titlerunning{Phylogenetic diversity and biodiversity indices on networks}        % if too long for running head

\author{Kristina Wicke \and Mareike Fischer}


\institute{K. Wicke \at Institute of Mathematics and Computer Science, Ernst-Moritz-Arndt University, Greifswald, Germany \\  \email{kristina.wicke@uni-greifswald.de}
           \and
           M. Fischer (\Letter) \at Institute of Mathematics and Computer Science, Ernst-Moritz-Arndt University, Greifswald, Germany \\ \email{email@mareikefischer.de}       
}

\date{Received: date / Accepted: date}

\maketitle

\begin{abstract}
The Shapley value, a solution concept from cooperative game theory, has recently been considered for both unrooted and rooted phylogenetic trees. Here, we focus on the Shapley value of unrooted trees and first revisit the so-called split counts of a phylogenetic tree and the Shapley transformation matrix that allows for the calculation of the Shapley value from the edge lengths of a tree. We show that non-isomorphic trees may have permutation-equivalent Shapley transformation matrices and permutation-equivalent null spaces. This implies that estimating the split counts associated with a tree or the Shapley values of its leaves does not suffice to reconstruct the correct tree topology. We then turn to the use of the Shapley value as a prioritization criterion in biodiversity conservation and compare it to a greedy solution concept. Here, we show that for certain phylogenetic trees, the Shapley value may fail as a prioritization criterion, meaning that the diversity spanned by the top $k$ species (ranked by their Shapley values) cannot approximate the total diversity of all $n$ species.
\end{abstract}

\keywords{Phylogenetic tree \and Shapley value \and Shapley transformation \and Noah's ark problem}

\section{Introduction}
The \emph{Shapley value} of phylogenetic trees has been frequently discussed as a prioritization tool in biodiversity conservation. 
It was first introduced by \citet{Haake2008} for unrooted phylogenetic trees, but has also been considered for rooted phylogenetic trees (cf. \citet{Hartmann2013,Fuchs2015, Wicke2017}). 
Here, we focus on unrooted phylogenetic trees and answer some of the questions posed in \citet{Haake2008}.
On the one hand, we consider the relationship between the tree topology, the so-called \emph{split counts} of a tree and the Shapley value. In particular, we show that non-isomorphic trees can have permutation-equivalent Shapley transformation matrices and, as a consequence, identical Shapley values. This means that estimating the Shapley values or split counts from data (and not inferring them from a tree) does not suffice to reconstruct the corresponding tree.

On the other hand, we consider the use of the Shapley value as a ranking criterion in the so-called \emph{Noah's ark problem} (\citet{Weitzman1998}) and compare it to a greedy solution concept (\citet{Steel2005}). We show that the Shapley value can perform very badly as a prioritization criterion for a certain class of phylogenetic trees. In fact, we show that the diversity of the top $k$ Shapley species (i.e., the $k$ species with the highest Shapley values), may not approximate the total diversity of all species at all, while the total diversity is well captured by the top $k$ greedy species (i.e., the $k$ species chosen by a greedy approach). 

The paper is organized as follows. After introducing some basic definitions and notations we turn to the Shapley transformation matrix of a phylogenetic tree and recall some known results. We then show that non-isomorphic trees can have permutation-equivalent Shapley transformation matrices. We conclude this paper by considering the Shapley value as a prioritization criterion in the Noah's ark problem (\citet{Weitzman1998}). 

\section{Preliminaries}
Let $T=(V(T),E(T))$ be a tree with nodes $V(T)$, edges $E(T)$, leaves $V_L \subseteq V(T)$ and no nodes of degree 2.
Let $X$ be a set of taxa and let $\phi: X \rightarrow V_L$ be a bijective mapping from the set of taxa into the set of leaves of $T$ ($X$ is therefore sometimes called \emph{leaf set}).
Then $\mathcal{T} \coloneqq (T,\phi)$ is called a \emph{phylogenetic $X$-tree} with \emph{treeshape/topology} $T$.
If all internal nodes are of degree $3$, we call $\mathcal{T}$ a \emph{binary phylogenetic $X$-tree}.
Without loss of generality we assume $X=\{1, \ldots, n\}$ and use $n = \vert X \vert$ to denote the number of leaves of a tree.
When we write $\vert \mathcal{T} \vert$ we also mean the number of leaves of the tree.
In biology, often \emph{rooted} phylogenetic trees with a designated root node (representing the last common ancestor of all present-day species) are considered, but here we will mostly be concerned with \emph{unrooted} phylogenetic trees. 
Note that throughout this paper we always mean unrooted binary phylogenetic trees when we refer to trees unless stated otherwise.
Moreover, we assume all edges in a tree to have positive edge lengths assigned to them (e.g., representing evolutionary time between speciation events or substitution rates) and denote the length of an edge $k \in E(T)$ by $\alpha_k$ (cf. Figure \ref{Fig1}).

Given a weighted unrooted phylogenetic tree and a subset $S \subseteq X$ of taxa, the \emph{phylogenetic diversity} $PD(S)$ of $S$ is defined as the sum of edge lengths in the smallest spanning tree that connects the taxa in $S$.

In the following, we will consider the \emph{phylogenetic tree game} introduced by \citet{Haake2008}, which is a cooperative game associated with a phylogenetic tree. Recall that in game theory, a \emph{cooperative game} is a pair $(N, v)$ consisting of a set of players $N=\{1, \ldots, n\}$ and a \emph{characteristic function} $v: 2^N \rightarrow \mathbb{R}$ that assigns a real number to every coalition $S \in 2^N$ of players. 
Given a phylogenetic tree $\mathcal{T}$ we define the \emph{phylogenetic tree game} as the pair $(X, PD_{\mathcal{T}})$ consisting of the set of species $X$ and the phylogenetic diversity measure $PD$ that assigns a real value to all subsets $S \subseteq X$ of species (adapted from \citet{Haake2008}).
An important solution concept in cooperative game theory is the so-called \emph{Shapley value} that can also be used in the context of phylogenetic tree games. Given a phylogenetic tree game $(X, PD_{\mathcal{T}})$, the Shapley value is the vector $SV = (SV_i)$ defined as
\begin{equation} \label{eq_sv}
SV_i(X, PD_{\mathcal{T}}) = \frac{1}{n!} \sum\limits_{\substack{S \subseteq X \\ i \in S}} \Big( (\vert S \vert -1)! \, (n - \vert S  \vert)! \, ( PD(S) - PD(S \setminus \{i\}) \Big),
\end{equation}
where $n = \vert X \vert$ and $S$ denotes a subset of species containing taxon $i$.
Biologically, the Shapley value of a given species may be interpreted as the average contribution of a species to overall phylogenetic diversity and thus has been suggested as a prioritization criterion in biodiversity conservation (cf. \citet{Haake2008}). 

Note that the Shapley value of a phylogenetic tree game is a linear function of the edge weights of the tree. This linear transformation is called the \emph{Shapley transformation} in \citet{Haake2008} and is the main focus of the following section.

\begin{figure}[htbp]
	\centering
	\includegraphics[scale=0.6]{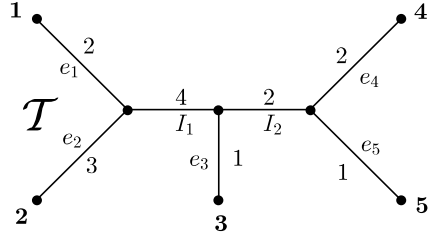}
	\caption{Phylogenetic $X$-tree $\mathcal{T}$ with leaf set $X=\{1,2,3,4,5\}$ and edge lengths $\alpha_{e_1} = 2, \, \alpha_{e_2}=3, \ldots, \alpha_{I_1}=4$ and $\alpha_{I_2}=2$. %
The vector of edge weights described in Definition \ref{shapleymatrix} is thus $(2,3,1,2,1,4,2)$.
	}
	\label{Fig1}
\end{figure}

\section{The Shapley Transformation Matrix}
Following the notation of \citet{Haake2008}, we refer to the weights of edges incident to leaves as \emph{leaf weights} and to the weights of internal edges as \emph{internal edge weights}. 
Recall that an unrooted binary phylogenetic tree on $n$ taxa has precisely $2n-3$ edges, whereof $n-3$ edges are internal edges (cf. \citet[p. 10]{Book_Steel}).
Then we can define the \emph{Shapley transformation matrix} as follows (taken from \citet{Haake2008}).

\begin{definition}[Shapley transformation matrix] \label{shapleymatrix}
Let $\mathcal{T}$ be a phylogenetic $X$-tree with leaf set $X=\{1, \ldots, n\}$, associated leaf weights $\alpha_1, \ldots, \alpha_n$ and internal edges $I_1, \ldots, I_{n-3}$ with associated internal edge weights $\alpha_{I_1}, \ldots, \alpha_{I_{n-3}}$. Let $\overrightarrow{E}$ be a vector consisting of the edge weights in this order: $(\alpha_1, \ldots, \alpha_n, \alpha_{I_1}, \ldots, \alpha_{I_{n-3}})^{\top}$. Then we define $\mathbf{M} = \mathbf{M}(X, \, PD_{\mathcal{T}})$ to be the $n \times (2n-3)$ matrix that corresponds to Equation \eqref{eq_sv} and therefore represents the Shapley transformation, such that the Shapley value of the game $(X, \, PD_{\mathcal{T}})$ is
$$ SV(X, PD_{\mathcal{T}}) = (SV_1, SV_2, \ldots, SV_n)^{\top} = \mathbf{M}\overrightarrow{E},$$
where $SV_i$ is the Shapley value of leaf $i$. 
The rows of $\mathbf{M}$ correspond to the leaves of the tree and the columns correspond to its edges.
\end{definition}

Note that the Shapley transformation matrix $\mathbf{M}$ depends on the tree topology. To be more precise, it was shown in \citet{Haake2008} that $\mathbf{M}$ depends on the so-called \emph{split counts} of a tree. 

\begin{definition}[Split counts]
Let $\mathcal{T}$ be a phylogenetic $X$-tree with leaf set $X=\{1, \ldots, n\}$ and edge set $E$. For a leaf $i \in X$ and an edge $k \in E$ the removal of $k$ splits $\mathcal{T}$ into two subtrees. Let $\mathcal{C}(i,k)$ denote the set of leaves in the subtree that contains $i$ (the \enquote{containing} subtree) and let $\mathcal{F}(i,k)$ denote the set of leaves in the other subtree that is \enquote{far} from $i$. We set $c(i,k) \coloneqq \vert \mathcal{C}(i,k) \vert$ and $f(i,k) \coloneqq \vert \mathcal{F}(i,k) \vert$ and call $c(i,k)$ and $f(i,k)$ the \emph{split counts} associated with leaf $i$ and edge $k$. Note that $c(i,k)+f(i,k)=n$ for all $i \in X$.
\end{definition}

\begin{example}
Consider leaf $3$ and edge $I_1$ of the phylogenetic tree $\mathcal{T}$ depicted in Figure \ref{Fig1}. Then $\mathcal{C}(3,I_1) = \{3,4,5\}$ and $\mathcal{F}(3, I_1)=\{1,2\}$. Thus, $c(3, I_1) = 3$ and $f(3, I_1) = 2$.
\end{example}

Based on the split counts of a phylogenetic tree the entries of the Shapley transformation matrix can be calculated as follows:
\begin{theorem}[\citet{Haake2008}] \label{theorem1}
Let $\mathcal{T}$ be a phylogenetic tree with $n$ leaves. Then the $(i,k)$th entry of the Shapley transformation matrix $\mathbf{M}$ is given by 
\begin{equation}
\mathbf{M}[i,k] = \frac{f(i,k)}{n \, c(i,k)}.
\end{equation}
\end{theorem}

The split counts can also be used to compute a basis for the null space of $\mathbf{M}$.
\begin{theorem}[\citet{Haake2008}]
Let $\mathcal{T}$ be a phylogenetic tree with leave set $X = \{1, \ldots, n\}$ and internal edges $I_1, \ldots, I_{n-3}$. The dimension of the null space of $\mathbf{M}=\mathbf{M}(X, PD_{\mathcal{T}})$ is $n-3$. 
A basis for the null space is the collection of vectors $\{w_{I_k}\}$ in $\mathbb{R}^{2n-3}$, one for each internal edge $I_k$:
	\begin{equation}
		\left(w_{I_k} \right)_i = \begin{cases}
			-  \frac{f(i,k)-1}{(n-2) \, c(i,k)}, & \text{ if } 1 \leq i \leq n\\
			1, & \text{ if } i=n+k \\
			0, & \text{ otherwise. }
		\end{cases}
	\end{equation}
\end{theorem}

\begin{example}
Consider the phylogenetic tree $\mathcal{T}$ on $n=5$ leaves depicted in Figure \ref{Fig1}. The Shapley transformation matrix for $\mathcal{T}$ is
$$ \mathbf{M}_{\mathcal{T}} = \bordermatrix{ & e_1 & e_2 & e_3 & e_4 & e_5 & I_1 & I_2 \cr \\[6pt]
1 & \nicefrac{4}{5} & \nicefrac{1}{20} & \nicefrac{1}{20} & \nicefrac{1}{20} & \nicefrac{1}{20} & \nicefrac{3}{10} & \nicefrac{2}{15} \cr \\[6pt]
2 & \nicefrac{1}{20} & \nicefrac{4}{5} & \nicefrac{1}{20} & \nicefrac{1}{20} & \nicefrac{1}{20} & \nicefrac{3}{10} & \nicefrac{2}{15} \cr \\[6pt]
3 & \nicefrac{1}{20} & \nicefrac{1}{20} & \nicefrac{4}{5} & \nicefrac{1}{20} & \nicefrac{1}{20} & \nicefrac{2}{15} & \nicefrac{2}{15} \cr \\ [6pt]
4 & \nicefrac{1}{20} & \nicefrac{1}{20} & \nicefrac{1}{20} & \nicefrac{4}{5} & \nicefrac{1}{20} & \nicefrac{2}{15} & \nicefrac{3}{10} \cr \\[6pt]
5 & \nicefrac{1}{20} & \nicefrac{1}{20} & \nicefrac{1}{20} & \nicefrac{1}{20} & \nicefrac{4}{5} & \nicefrac{2}{15} & \nicefrac{3}{10}} $$
and the Shapley value for the game $(X, PD_{\mathcal{T}})$ calculates as
$$ SV(X, PD_{\mathcal{T}})  = \mathbf{M} \cdot \overrightarrow{E}^{\top} = \mathbf{M} \cdot (2,3,1,2,1,4,2)^{\top} = \left( \nicefrac{41}{12}, \nicefrac{25}{6}, 2, \nicefrac{37}{12}, \nicefrac{7}{3} \right)^{\top}.$$
A basis for the null space is
$$ \left\lbrace \begin{pmatrix}
- \nicefrac{1}{3} \\ - \nicefrac{1}{3} \\ - \nicefrac{1}{9} \\ - \nicefrac{1}{9} \\ - \nicefrac{1}{9} \\ 1 \\ 0
\end{pmatrix},
\begin{pmatrix}
- \nicefrac{1}{9} \\ - \nicefrac{1}{9} \\ - \nicefrac{1}{9} \\ - \nicefrac{1}{3}  \\ - \nicefrac{1}{3} \\ 0 \\ 1
\end{pmatrix}
 \right\rbrace.$$  \\
\end{example}

%Note that since the split counts of a tree only depend on the topolgy of the tree, so do the Shapley transformation matrix $\mathbf{M}$ and the null space of $\mathbf{M}$. 
\noindent Moreover, following \citet{Haake2008} we call two trees \emph{isomorphic} if there is a bijection between the edges that maps one tree to the other and preserves the topological structure of the tree. 
Note that here we are only taking into account the treeshape or topology and not the labeling of leaves, i.e., we for example regard $\mathcal{T}$ and $\mathcal{T}'$ depicted in Figure \ref{Fig2} as isomorphic, because they have the same topology. Still, they depict different evolutionary relationships between the species $1,2,3$ and $4$. \\

\begin{figure}
	\centering
	\includegraphics[scale=0.9]{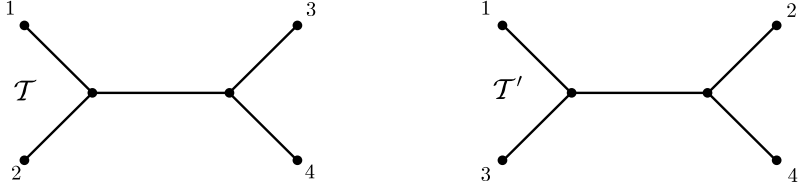}
	\caption{Phylogenetic $X$-trees $\mathcal{T}$ and $\mathcal{T}'$ on $X = \{1,2,3,4\}$ that are regarded as \emph{isomorphic}, because they have the same topology.}
	\label{Fig2}
\end{figure}

\noindent We call two matrices \emph{permutation-equivalent} if they only differ by a permutation of the rows and a permutation of the columns, i.e., two matrices $M_1$ and $M_2 \, \in \mathbb{R}^{m \times n}$ are permutation-equivalent if there exists a permutation matrix $P \in \mathbb{R}^{m \times m}$ and a permutation matrix $Q \in \mathbb{R}^{n \times n}$ such that
$$ P \, M_1 \, Q = M_2.$$
Similarly, we call two subspaces of $\mathbb{R}^n$ \emph{permutation-equivalent} if one space can be obtained from the other by some permutation of the coordinates. Based on this notation we can restate the following theorem from \citet{Haake2008}.

\begin{theorem}[\citet{Haake2008}] \label{theorem3}
Isomorphic trees induce permutation-equivalent Shapley transformation matrices with permutation-equivalent null spaces. Hence, if for two trees $\mathcal{T}_1, \mathcal{T}_2$, their Shapley transformation matrices $\mathbf{M}_1, \mathbf{M}_2$ or their null spaces are not permutation-equivalent, then $\mathcal{T}_1, \mathcal{T}_2$ must not be isomorphic.
\end{theorem}

Theorem \ref{theorem3} follows from the fact that the split counts of a tree only depend on the topological structure of the tree. To be precise, isomorphic trees induce the same Shapley transformation matrix $\mathbf{M}$ up to a permutation of the rows (given by permuting the order of the leaves that define the rows) and a permutation of the columns (given by permuting the order of the edges that define  the columns). 
Note that this means that isomorphic trees also induce permutation-equivalent null spaces, because while the null space of $\mathbf{M}$ is not affected by permuting the rows of $\mathbf{M}$, a permutation of the columns of $\mathbf{M}$ induces a permutation of the coordinates of the null space of $\mathbf{M}$.
\\

In their paper, \citet{Haake2008} raise two questions concerning the relationship between the split counts of a tree, its topology and the Shapley transformation matrix, namely:

\begin{enumerate}
	\item Is there a way to determine or estimate split counts from data, and can this assist in determining the correct tree topology?
	\item Does the converse of Theorem \ref{theorem3} hold, i.e., if two trees have permutation-equivalent Shapley transformation matrices or permutation-equivalent null spaces, are they isomorphic?
\end{enumerate}

In the following, we present our main result. We show that there are non-isomorphic trees, i.e., trees of different topology, that induce permutation-equivalent Shapley transformation matrices and permutation-equivalent null spaces. This implies that we can negate the second question. We then also negate the second part of the first question, because our results show that split counts are not sufficient to determine the topology of a tree.

\begin{theorem} \label{observation1}
Two trees $\mathcal{T}_1, \mathcal{T}_2$ with permutation-equivalent Shapley transformation matrices or permutation-equivalent null spaces are not necessarily isomorphic.
\end{theorem}

\begin{proof}
Consider the two trees $\mathcal{T}_1$ and $\mathcal{T}_2$ depicted in Figure \ref{Fig3}. $\mathcal{T}_1$ and $\mathcal{T}_2$ are clearly not isomorphic, because they have different topologies. However, for each leaf $i \in \{1, \ldots, 16\}$ and edge $k \in \{e_1, \ldots, e_{16}, I_1, \ldots, I_{13}\}$ both trees exhibit the same split counts. Thus, they induce permutation-equivalent Shapley transformation matrices and permutation-equivalent null spaces. Note that in this case the Shapley transformation matrices and null spaces are not only permutation-equivalent, but in fact identical.
\qed
\end{proof}

\begin{figure} [htbp] 
	\centering
	\includegraphics[scale=0.9]{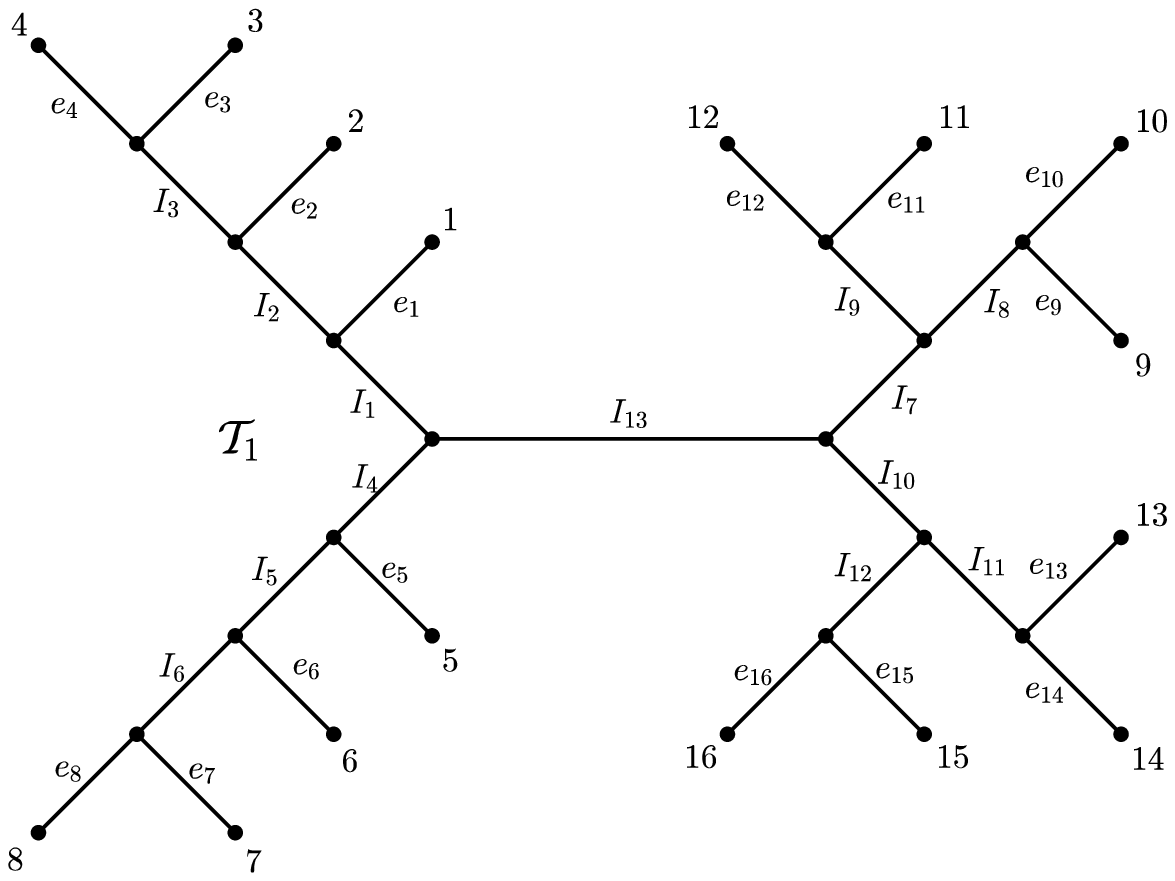} \\
	\includegraphics[scale=0.9]{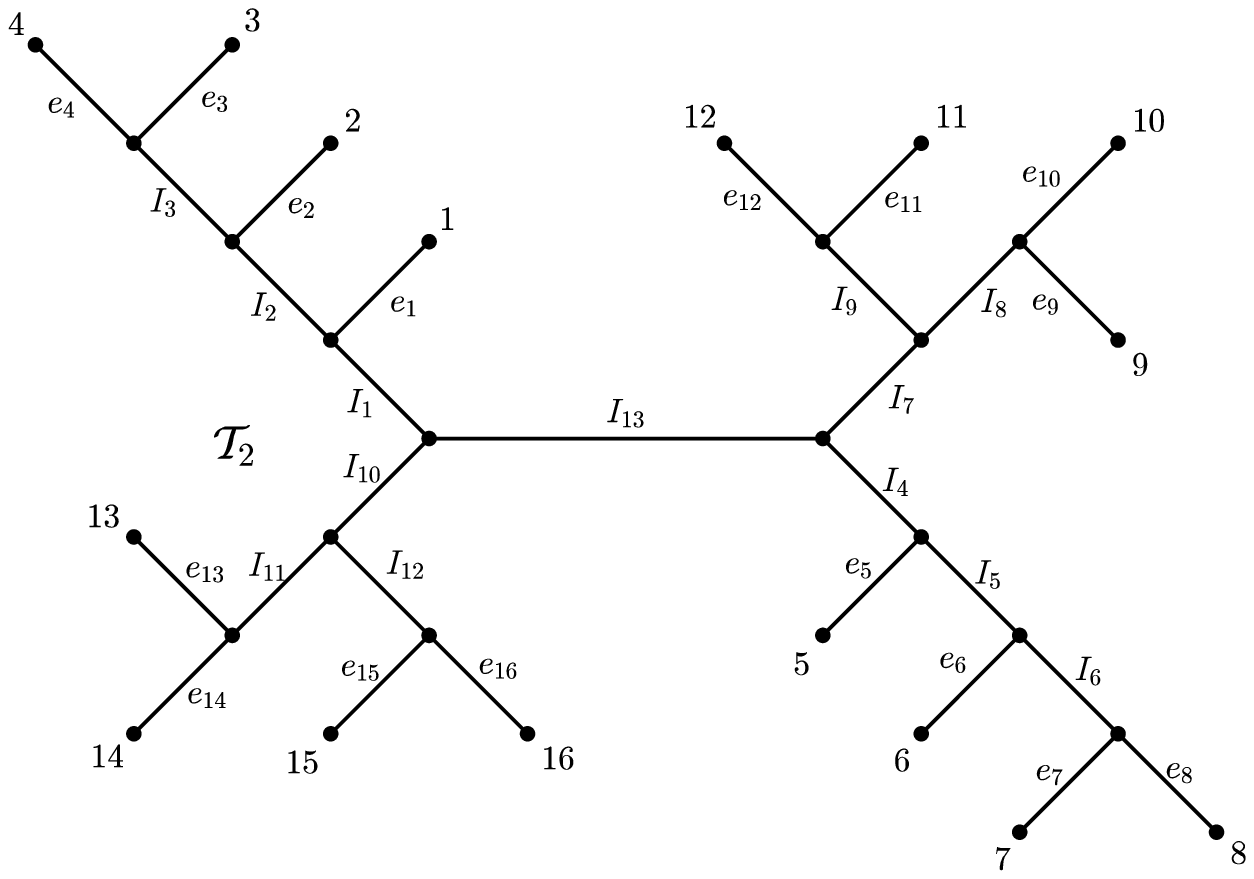}
	\caption{Two non-isomorphic trees $\mathcal{T}_1$ and $\mathcal{T}_2$ on 16 leaves that induce permutation-equivalent Shapley transformation matrices and permutation-equivalent null spaces.}
	\label{Fig3}
\end{figure}

\begin{remark}
The pair $(\mathcal{T}_1, \mathcal{T}_2)$ in Figure \ref{Fig3} is the smallest example for a pair of non-isomorphic trees inducing permutation-equivalent Shapley transformation matrices and permutation-equivalent null spaces, which we verified by an exhaustive search and analysis of all tree topologies on fewer than 16 leaves and their split counts. \\
Note that if we had also taken into account the labeling of leaves and not only the tree topology when defining  isomorphism of trees, i.e., if we had not regarded $\mathcal{T}$ and $\mathcal{T}'$ (Figure \ref{Fig2}) as isomorphic, but as non-isomorphic, then $\mathcal{T}$ and $\mathcal{T}'$ would have been a   smallest example, because clearly they induce permutation-equivalent Shapley transformation matrices and permutation-equivalent null spaces (since they share the same topology).

However, in the following we generalize the pair $(\mathcal{T}_1, \mathcal{T}_2)$ to a class of pairs $(\mathcal{T}_1^{\ast}, \mathcal{T}_2^{\ast})$, where $\mathcal{T}_1^{\ast}$ and $\mathcal{T}_2^{\ast}$ are non-isomorphic trees on $\geq 16$ leaves that induce permutation-equivalent Shapley transformation matrices and permutation-equivalent null spaces.
\end{remark}

\begin{figure}[htbp]
	\centering
	\includegraphics[scale=0.75]{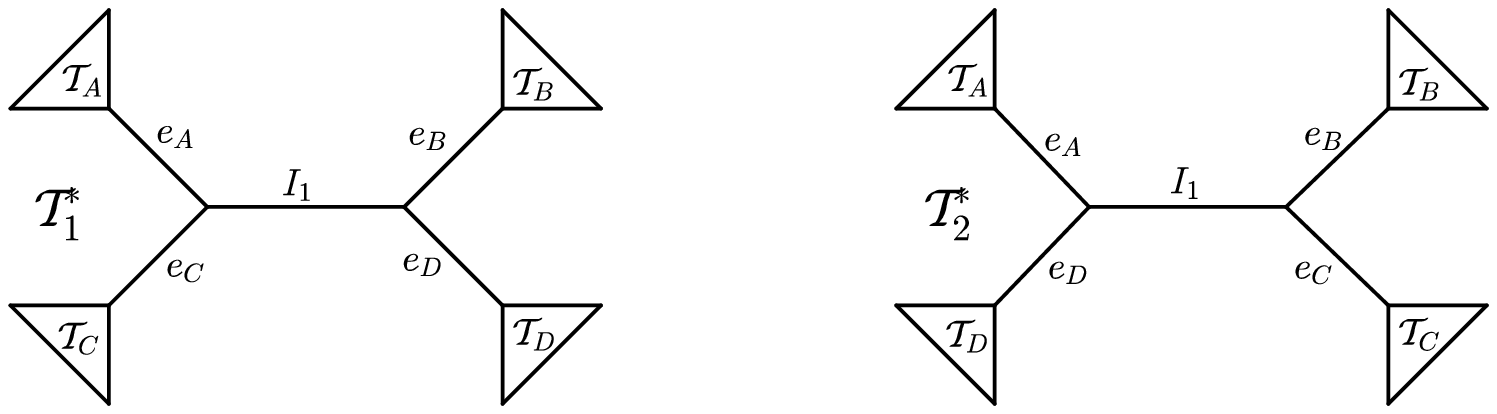}
	\caption{Two non-isomorphic trees $\mathcal{T}_1^{\ast}$ and $\mathcal{T}_2^{\ast}$ consisting of four rooted subtrees, where $\left\vert \mathcal{T}_A \right\vert = \left\vert \mathcal{T}_B \right\vert$, but $\mathcal{T}_A$ and $\mathcal{T}_B$ have different topologies and $\left\vert \mathcal{T}_C \right\vert = \left\vert \mathcal{T}_D \right\vert$, but $\mathcal{T}_C$ and $\mathcal{T}_D$ have different topologies ($ \left\vert \mathcal{T}_A \right\vert, \, \left\vert \mathcal{T}_B \right\vert, \, \left\vert \mathcal{T}_C \right\vert$ and $\left\vert \mathcal{T}_D \right\vert \, \geq 4$).}
	\label{Fig4}
\end{figure}

\begin{theorem} \label{observation2}
Trees of type $\mathcal{T}_1^{\ast}$ and $\mathcal{T}_2^{\ast}$ as in Figure \ref{Fig4} induce permutation-equivalent Shapley transformation matrices and permutation-equivalent null spaces, but are not isomorphic.
\end{theorem}

\begin{remark}
Setting $\mathcal{T}_A, \mathcal{T}_C$ to the so-called rooted caterpillar tree on four leaves and $\mathcal{T}_B, \mathcal{T}_D$ to the fully balanced tree on four leaves, $\mathcal{T}_1^{\ast}$ and $\mathcal{T}_1^{\ast}$ coincide with $\mathcal{T}_1$ and $\mathcal{T}_2$ depicted in Figure \ref{Fig3} and used in the proof of Theorem \ref{observation1}.
\end{remark}

\begin{proof} 
Let $\mathcal{T}_1^{\ast}$ and $\mathcal{T}_2^{\ast}$ be two trees consisting of four rooted subtrees as depicted in Figure \ref{Fig4}, where $S$ denotes the leaf set of $\mathcal{T}_s$ with $s \in \{A,B,C,D\}$ and where 
	\begin{itemize}
	\item $\left\vert A \right\vert = \left\vert B \right\vert, \, $ but $\mathcal{T}_A$ and $\mathcal{T}_B$ are of different shape, 
	\item $\left\vert C \right\vert = \left\vert D \right\vert, \,$ but $\mathcal{T}_C$ and $\mathcal{T}_D$ are of different shape. 
	\end{itemize}
Note that this implies $\left\vert A \right\vert, \, \left\vert B \right\vert, \, \left\vert C \right\vert$ and $\left\vert D \right\vert \geq 4$ (cf. \citet[p. 25]{Semple2003}).
Clearly, $\mathcal{T}_1^{\ast}$ and $\mathcal{T}_2^{\ast}$ are not isomorphic, since both $\mathcal{T}_
	A,\mathcal{T}_B$ and $\mathcal{T}_C, \mathcal{T}_D$ are of different shapes. 
In order to show that $\mathcal{T}_1^{\ast}$ and $\mathcal{T}_2^{\ast}$ induce permutation-equivalent Shapley transformation matrices and permutation-equivalent null spaces we show that they exhibit the same split counts for each leaf $i$ and edge $k$. Here, we distinguish between different cases:
	\begin{itemize}
	%Leaf and edge in the same subtree
	\item Leaf $i$ and edge $k$ are part of the same subtree. Let $S$ denote the leaf set of $\mathcal{T}_s$ with $s \in \{A,B,C,D\}$. 
	Exemplarily we assume that $k$ and $i$ are in $\mathcal{T}_A$. Then $k$ induces a split $A_1 \vert A_2$ in $\mathcal{T}_A$ (and thus in $\mathcal{T}_1^{\ast}$ and $\mathcal{T}_2^{\ast}$) with $A_1, A_2 \subset A, \, A_1 \cap A_2 = \emptyset,$ and $A_1 \cup A_2 = A$. Without loss of generality let $A_1$ be the set of leaves of $\mathcal{T}_A$ that is still connected to the rest of $\mathcal{T}_1^{\ast}$, respectively $\mathcal{T}_2^{\ast}$, while $A_2$ is the set that is far from it.
	Now, there are two cases:
	\begin{enumerate}
		\item Leaf $i$ is in $A_1$. Then 
			\begin{align*}
			c(i,k) &= \vert A_1 \vert + \vert B \vert + \vert C \vert + \vert D \vert \\
			f(i,k) &= \vert A_2 \vert.
			\end{align*}
		\item Leaf $i$ is in $A_2$. Then 
			\begin{align*}
			c(i,k) &= \vert A_2 \vert \\
			f(i,k) &= \vert A_1 \vert + \vert B \vert + \vert C \vert + \vert D \vert
			\end{align*}
    \end{enumerate}	
However, this holds both for $\mathcal{T}_1^{\ast}$ and in $\mathcal{T}_2^{\ast}$.
Thus, the split counts $f(i,k)$ and $c(i,k)$ induced by $i$ and $k$ are the same in $\mathcal{T}_1^{\ast}$ and in $\mathcal{T}_2^{\ast}$.
Analogously, this follows if $i$ and $k$ are in $\mathcal{T}_B, \mathcal{T}_C$ or $\mathcal{T}_D$.
%%%
%Leaf and edge in different subtrees
\item Leaf $i$ and edge $k$ are part of different subtrees. Again, let $S$ denote the leaf set of $\mathcal{T}_s$ with  $s \in \{A,B,C,D\}$. Exemplarily we assume that leaf $i$ is in $\mathcal{T}_A$ and edge $k$ is in $\mathcal{T}_B$. Then $k$ induces a split $B_1 \vert B_2$ in $\mathcal{T}_B$ (and thus in $\mathcal{T}_1^{\ast}$ and $\mathcal{T}_2^{\ast}$) with $B_1, B_2 \subset B, \, B_1 \cap B_2 = \emptyset,$ and $B_1 \cup B_2 = B$. Without loss of generality let $B_1$ be the set of leaves of $\mathcal{T}_B$ that is still connected to the rest of $\mathcal{T}_1^{\ast}$, respectively $\mathcal{T}_2^{\ast}$, while $B_2$ is the set that is far from it. Then both in $\mathcal{T}_1^{\ast}$ and in $\mathcal{T}_2^{\ast}$ we have
	\begin{align*}
	c(i,k) &= \vert A \vert + \vert B_1 \vert + \vert C \vert + \vert D \vert \\
	f(i,k) &= \vert B_2 \vert.
    \end{align*}
Thus, both trees exhibit the same split counts induced by leaf $i$ and edge $k$. Analogously this follows for all other cases where leaf $i$ and edge $k$ are in different subtrees.
%%%%
%Edges leading to subtrees
\item Now consider the edges $e_A, \ldots, e_D$ as depicted in Figure \ref{Fig4}. Let $S$ denote the leaf set of $\mathcal{T}_s$ with $s \in \{A,\ldots, D\}$. Then both in $\mathcal{T}_1^{\ast}$ and $\mathcal{T}_2^{\ast}$, $e_A$ induces the split $A \vert BCD$, $e_B$ induces the split $B \vert ACD$ and so forth. Exemplarily, we consider $e_A$ and distinguish between two cases:
		\begin{enumerate}
		\item Leaf $i \in A$. Then $c(i,e_A) = \vert A \vert$ and $f(i,e_A) = \vert B \vert + \vert C \vert + \vert D \vert$. 
		\item Leaf $i \notin A$. Then $c(i,e_A) = \vert B \vert + \vert C \vert + \vert D \vert$ and $f(i,e_A) = \vert A \vert$.
		Again, this holds both in $\mathcal{T}_1^{\ast}$ and in $\mathcal{T}_2^{\ast}$, thus the split counts induced by a leaf $i$ and edge $e_A$ are the same. 
		\end{enumerate}
		Analogously, this follows for $e_B, e_C$ and $e_D$.
%%%
%Edge I_1
\item Now consider edge $I_1$. Again, we use $S$ to denote the leaf set of $\mathcal{T}_s$ with $s \in \{A, \ldots, D\}$. In $\mathcal{T}_1^{\ast}$, $I_1$ induces the split $AC \vert BD$, while in $\mathcal{T}_2^{\ast}$ it induces the split $AD \vert BC$. Recall that by assumption $\vert A \vert = \vert B \vert$ and $\vert C \vert = \vert D \vert$.
	\begin{enumerate}
	\item Leaf $i \in A$:
		\begin{itemize}
		\item Split counts in $\mathcal{T}_1^{\ast}$:
			\begin{align*}
			c(i, I_1) = \vert A \vert + \vert C \vert = \vert A \vert + \vert D \vert \\
			f(i, I_1) = \vert B \vert + \vert D \vert = \vert B \vert + \vert C \vert
			\end{align*}
		\item Split counts in $\mathcal{T}_2^{\ast}$:
			\begin{align*}
			c(i, I_1) = \vert A \vert + \vert D \vert \\
			f(i, I_1) = \vert B \vert + \vert C \vert
			\end{align*}
		\end{itemize}
	\item Leaf $i \in B$:
		\begin{itemize}
		\item Split counts in $\mathcal{T}_1^{\ast}$:
			\begin{align*}
			c(i, I_1) = \vert B \vert + \vert D \vert = \vert B \vert + \vert C \vert \\
			f(i, I_1) = \vert A \vert + \vert C \vert = \vert A \vert + \vert D \vert
			\end{align*}
		\item Split counts in $\mathcal{T}_2^{\ast}$:
			\begin{align*}
			c(i, I_1) = \vert B \vert + \vert C \vert \\
			f(i, I_1) = \vert A \vert + \vert D \vert
			\end{align*}
		\end{itemize}
		\item Leaf $i \in C$:
		\begin{itemize}
		\item Split counts in $\mathcal{T}_1^{\ast}$:
			\begin{align*}
			c(i, I_1) = \vert A \vert + \vert C \vert = \vert B \vert + \vert C \vert \\
			f(i, I_1) = \vert B \vert + \vert D \vert = \vert A \vert + \vert D \vert
			\end{align*}
		\item Split counts in $\mathcal{T}_2^{\ast}$:
			\begin{align*}
			c(i, I_1) = \vert B \vert + \vert C \vert \\
			f(i, I_1) = \vert A \vert + \vert D \vert
			\end{align*}
		\end{itemize}
		\item Leaf $i \in D$:
		\begin{itemize}
		\item Split counts in $\mathcal{T}_1^{\ast}$:
			\begin{align*}
			c(i, I_1) = \vert B \vert + \vert D \vert = \vert A \vert + \vert D \vert \\
			f(i, I_1) = \vert A \vert + \vert C \vert = \vert B \vert + \vert C \vert
			\end{align*}
		\item Split counts in $\mathcal{T}_2^{\ast}$:
			\begin{align*}
			c(i, I_1) = \vert A \vert + \vert D \vert \\
			f(i, I_1) = \vert B \vert + \vert C \vert
			\end{align*}
		\end{itemize}
	\end{enumerate}
	Thus, in all cases the split counts induced by edge $I_1$ and any leaf coincide in $\mathcal{T}_1^{\ast}$ and $\mathcal{T}_2^{\ast}$.
\end{itemize}
Since $\mathcal{T}_1^{\ast}$ and $\mathcal{T}_2^{\ast}$ exhibit the same split counts $c(i,k)$ and $f(i,k)$ for each leaf $i$ and edge $k$, they induce permutation-equivalent Shapley transformation matrices and permutation-equivalent null spaces.
\qed
\end{proof}

Note that the above Theorems (Theorems \ref{observation1} and \ref{observation2}) show that determining or estimating split counts from data cannot assist in determining the correct tree topology, because non-isomorphic trees may exhibit identical split counts. 
Neither does estimating the Shapley value from data assist in determining the correct tree topology, because non-isomorphic trees may also have identical Shapley values. Consider for example $\mathcal{T}_1$ and $\mathcal{T}_2$ depicted in Figure \ref{Fig3} and set all edge lengths to one. Since $\mathcal{T}_1$ and $\mathcal{T}_2$ have permutation-equivalent Shapley transformation matrices, the Shapley values of the leaves in $\mathcal{T}_1$ and $\mathcal{T}_2$ coincide.

Thus, we conclude this section with another theorem.
\begin{theorem}
Neither split counts nor the Shapley values of all leaves (e.g. estimated from data) suffice to reconstruct the correct tree topology.
\end{theorem}

\section{The Shapley Value and the Noah's Ark Problem}
We now turn to an application of the Shapley value of phylogenetic trees, namely its use as a criterion for prioritizing species in nature conservation. In particular, we consider a simple variant of the so-called \emph{Noah's ark problem} (NAP) (cf. \citet{Weitzman1998}) and compare the Shapley value to a greedy approach introduced by \citet{Steel2005}. 
To be precise, we look for a subset $W$ of $X$ of given size, say $k$, that has maximal $PD$ score. 
In other words, given a phylogenetic $X$-tree $\mathcal{T}$ and $k \in \mathbb{N}$, we look for a maximum-weight subtree of $\mathcal{T}$ with $k$ leaves. 
\citet{Steel2005} showed that a greedy algorithm can solve this problem. 
For $k \geq 1$, let 
$$ pd_k = \max \{ PD(W): \, W \subseteq X, \, \vert W \vert = k \}$$
denote the largest possible phylogenetic diversity value across all subsets of species of size $k$ and let
$$ PD_k = \{ W \subseteq X: \,  \vert W \vert = k \text{ and } PD(W) = pd_k\}$$
be the set of all collections of $k$ species that realize this maximal phylogenetic diversity (taken from \citet{Steel2005}). Then a greedy algorithm can be used to determine $PD_k$.

\begin{theorem}[\citet{Steel2005}] \label{theorem4}
$PD_k$ consists precisely of those subsets of $X$ of size $k$ that can be built up as follows: 
Select any pair of species that are maximally far apart (in the edge-weighted tree $\mathcal{T}$) and then sequentially add elements of $X$ so as to maximize at each step the increase in $PD$ score.
\end{theorem}

\noindent \citet{Haake2008} now state the following question:
\begin{itemize}
	\item If we use the Shapley value to rank species in the Noah's ark problem for preservation, to what extent can we guarantee that the diversity of the top $k$ species (i.e., the weight of the subtree spanning them) approximates the total diversity of all $n$ species? 
\end{itemize}

\noindent In the following, we show that for certain trees the diversity of the top $k$ species (ranked by their Shapley values) tends to zero, while the diversity of all $n$ species tends to infinity. Thus, the top $k$ species cannot approximate the total diversity of all $n$ species.

\begin{figure}[htbp]
\centering
\includegraphics[scale=0.7]{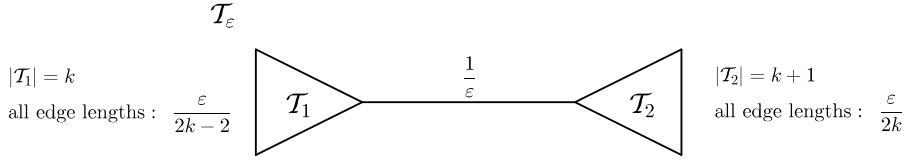}
\caption{Phylogenetic tree $\mathcal{T}_{\varepsilon}$ on $n=2k+1$ leaves consisting of two subtrees $\mathcal{T}_1$ and $\mathcal{T}_2$ with $\vert \mathcal{T}_1 \vert = k$ and $\vert \mathcal{T}_2 \vert = k+1$. The edge connecting $\mathcal{T}_1$ and $\mathcal{T}_2$ has length $\frac{1}{\varepsilon}$, while all edges in $\mathcal{T}_1$ have length $\frac{\varepsilon}{2k-2}$ and all edges in $\mathcal{T}_2$ have length $\frac{\varepsilon}{2k}$.}
\label{Fig5}
\end{figure}

\begin{theorem}
Let $\mathcal{T}_{\varepsilon}$ be a phylogenetic tree on $n=2k+1$ leaves consisting of two subtrees $\mathcal{T}_1$ and $\mathcal{T}_2$ with $\vert \mathcal{T}_1 \vert = k$ and $\vert \mathcal{T}_2 \vert = k+1$ and $k \geq 1$ as depicted in Figure \ref{Fig5}. Let the edge connecting $\mathcal{T}_1$ and $\mathcal{T}_2$ have length $\frac{1}{\varepsilon}$, and let all edges in $\mathcal{T}_1$ have length $\frac{\varepsilon}{2k-2}$ and all edges in $\mathcal{T}_2$ have length $\frac{\varepsilon}{2k}$.
Moreover, let 
$$ 0 < \varepsilon < \sqrt{\frac{k+2}{k^3+3k^2-2}}.$$
Then the top $k' \leq k$ species (ranked by their Shapley values) are all in $\mathcal{T}_1$, and for $\varepsilon \rightarrow 0$ their diversity tends to zero, while the diversity of all $n$ species tends to infinity. 
\end{theorem}

\begin{proof}
Let $\mathcal{T}_{\varepsilon}$ be as depicted in Figure \ref{Fig5}. 
Let the leaf set of $\mathcal{T}_1$ be $\{s_1^1, s_2^1, \ldots, s_k^1\}$ and let the leaf set of $\mathcal{T}_2$ be $\{s_1^2, s_2^2, \ldots, s_{k+1}^2\}$ with $k \geq 1$.
Let $s_1, \ldots, s_{k'}$ be the top $k'$ species (ranked by their Shapley values) with $k' \leq k$. 
In the following, we will show that for $ 0 < \varepsilon < \sqrt{\frac{k+2}{k^3+3k^2-2}}$ the species $s_1, \ldots, s_{k'}$ are all leaves of $\mathcal{T}_1$.
Note that $\mathcal{T}_1$ is a rooted phylogenetic tree on $k$ leaves and $\mathcal{T}_2$ is a rooted phylogenetic tree on $k+1$ leaves. Thus, $\mathcal{T}_1$ has $2k-2$ edges and $\mathcal{T}_2$ has $2(k+1)-2 = 2k$ edges (cf. \citet[p. 10]{Book_Steel}).
Thus, the diversity spanned by all leaves of $\mathcal{T}_1$ calculates as
$$ PD(\{s_1^1, s_2^1, \ldots, s_k^1\}) = (2k-2) \cdot \frac{\varepsilon}{2k-2} = \varepsilon,$$
because all edge lengths in $\mathcal{T}_1$ are defined as $\frac{\varepsilon}{2k-2}$.
If we can show that the top $k'$ species $s_1, \ldots, s_{k'}$ (ranked by their Shapley values) are all leaves of $\mathcal{T}_1$ this implies
$$ 0 \leq  PD (\{s_1, \ldots, s_{k'} \}) \leq \varepsilon \, \overset{\varepsilon \rightarrow 0}{\longrightarrow} \, 0.$$
Considering the diversity of all $n=2k+1$ species, however, we have
\begin{align*}
PD(\{s_1^1, s_2^1, \ldots, s_k^1, s_1^2, s_2^2, \ldots, s_{k+1}^2\})  &= (2k-2) \cdot \frac{\varepsilon}{2k-2} + \frac{1}{\varepsilon} + 2k \cdot \frac{\varepsilon}{2k} \\
&= 2 \varepsilon + \frac{1}{\varepsilon} \overset{\varepsilon \rightarrow 0}{\longrightarrow} \, \infty.
\end{align*}
Thus, it remains to show that the top $k'$ species (ranked by their Shapley values) are all in $\mathcal{T}_1$, i.e., we need to show that $\{s_1, \ldots, s_k'\} \subseteq \{s_1^1, \ldots, s_k^1\}$.
The idea is to show that the minimal Shapley value of any species in $\mathcal{T}_1$ is still greater than the maximal Shapley value of any species in $\mathcal{T}_2$. Thus, we define
\begin{align*}
SV_{\mathcal{T}_1}^{\min} &\coloneqq \min\limits_{s_i^1 \, \in \, \{s_1^1, \ldots, s_k^1\}} SV(s_i^1) \, \text{ and } \\
SV_{\mathcal{T}_2}^{\max} &\coloneqq \max\limits_{s_j^2 \, \in \, \{s_1^2, \ldots, s_{k+1}^2\}} SV(s_j^2).
\end{align*}
We now develop bounds for $SV_{\mathcal{T}_1}^{\min}$ and $SV_{\mathcal{T}_2}^{\max}$ and then show that if
$0 < \varepsilon < \sqrt{\frac{k+2}{k^3+3k^2-2}}$, we have $ SV_{\mathcal{T}_1}^{\min} > SV_{\mathcal{T}_2}^{\max}$.
Recall that the Shapley value of a phylogenetic tree game can be calculated by multiplying the Shapley transformation matrix $\mathbf{M}$ with the vector of edge lengths $\overrightarrow{E}$ of the tree, where the $(i,k)$th entry of the Shapley transformation matrix is given by $\mathbf{M}[i,k] = \frac{f(i,k)}{n \, c(i,k)}$ which we call \emph{split factor} in the following (cf. Theorem \ref{theorem1}). 
We now develop a lower bound for $SV_{\mathcal{T}_1}^{\min}$ by considering bounds on $f(i,k)$ and $c(i,k)$.
Let $e$ denote the edge connecting $\mathcal{T}_1$ and $\mathcal{T}_2$. Then we derive the following lower bound for $SV_{\mathcal{T}_1}^{\min}$:
\begin{align*}
SV_{\mathcal{T}_1}^{\min} &\geq \underbrace{\underbrace{\frac{k+1}{(2k+1) k}}_{\text{split factor}} \cdot \underbrace{\frac{1}{\varepsilon}}_{\text{length of edge } e}}_{\text{ contribution of } e} + \underbrace{\underbrace{(2k-2)}_{\text{number of edges in } \mathcal{T}_1} \cdot \underbrace{\frac{1}{(2k+1) 2k}}_{\text{split factor}} \cdot \underbrace{\frac{\varepsilon}{2k-2}}_{\text{length of an edge in }  \mathcal{T}_1}}_{\text{contribution of edges in } \mathcal{T}_1}   \\
&\hspace*{5mm} + \underbrace{\underbrace{2k}_{\text{number of edges in } \mathcal{T}_2} \cdot \underbrace{\frac{1}{(2k+1)2k}}_{\text{split factor}} \cdot \underbrace{\frac{\varepsilon}{2k}}_{\text{length of an edge in } \mathcal{T}_2}}_{\text{contribution of edges in } \mathcal{T}_2} \\
&= \frac{k+1}{\varepsilon (2k^2+k)} + \frac{\varepsilon}{4k^2+2k} + \frac{\varepsilon}{4k^2+2k} \\
&= \frac{k+1}{\varepsilon (2k^2+k)} + \frac{\varepsilon}{2k^2+k}.
\end{align*}
Here, for edge $e$ separating $\mathcal{T}_1$ and $\mathcal{T}_2$ and any leaf $i$ in $\mathcal{T}_1$, we have $f(i,e) = k+1$ and $c(i,e)=k$, because there are $k+1$ leaves in $\mathcal{T}_2$ (that is \enquote{far} from $\mathcal{T}_1$) and $\mathcal{T}_1$ has $k$ leaves. For any edge $e'$ in $\mathcal{T}_1$ or $\mathcal{T}_2$ and any leaf $i$ in $\mathcal{T}_1$, we have $f(i,e') \geq 1$ and $c(i,e') \leq 2k$, because $e'$ separates at least one leaf from the rest of the tree.
Similarly, we develop an upper bound for $SV_{\mathcal{T}_2}^{\max}$.
\begin{align*}
SV_{\mathcal{T}_2}^{\max} &\leq \underbrace{\underbrace{\frac{k}{(2k+1)(k+1)}}_{\text{split factor}} \cdot \underbrace{\frac{1}{\varepsilon}}_{\text{length of edge } e}}_{\text{ contribution of } e} + \underbrace{\underbrace{(2k-2)}_{\text{number of edges in } \mathcal{T}_1} \cdot \underbrace{\frac{k-1}{(2k+1)(k+2)}}_{\text{split factor}} \cdot \underbrace{\frac{\varepsilon}{2k-2}}_{\text{length of an edge in }  \mathcal{T}_1}}_{\text{contribution of edges in } \mathcal{T}_1}   \\
&\hspace*{5mm} + \underbrace{\underbrace{2k}_{\text{number of edges in } \mathcal{T}_2} \cdot \underbrace{\frac{2k}{(2k+1) \cdot 1}}_{\text{split factor}} \cdot \underbrace{\frac{\varepsilon}{2k}}_{\text{length of an edge in } \mathcal{T}_2}}_{\text{contribution of edges in } \mathcal{T}_2} \\
&= \frac{k}{\varepsilon (2k^2+3k+1)} + \frac{\varepsilon (k-1)}{2k^2+5k+2} + \frac{\varepsilon 2k}{2k+1} \\
&= \frac{k}{\varepsilon (2k^2+3k+1)} + \frac{\varepsilon (2k^2+5k-1)}{2k^2+5k+2}.
\end{align*}
Here, for edge $e$ and any leaf $j$ in $\mathcal{T}_2$, we have $f(j,e)=k$ and $c(j,e)=k+1$, because $e$ separates $\mathcal{T}_1$ from $\mathcal{T}_2$. For an edge $e_1$ in $\mathcal{T}_1$ and any leaf $j$ in $\mathcal{T}_2$, $e_1$ separates at most $k-1$ leaves in $\mathcal{T}_1$ from $j$, thus $f(j,e_1) \leq k-1$ and $c(j,e_1) \geq k+2$. Conversely, for an edge $e_2$ in $\mathcal{T}_2$ and any leaf $j$ in $\mathcal{T}_2$, $e_2$ separates at most $k$ leaves of $\mathcal{T}_1$ and $k$ leaves of $\mathcal{T}_2$ from $j$. Thus, $f(j,e_2) \leq 2k$ and $c(j,e_2) \geq 1$.
Now, we compare the lower bound for $SV_{\mathcal{T}_1}^{\min}$ and the upper bound for $SV_{\mathcal{T}_2}^{\max}$. Using Mathematica (\citet{Mathematica}) we solved the inequality
\begin{equation*} \label{inequality}
 \frac{k+1}{\varepsilon (2k^2+k)} + \frac{\varepsilon}{2k^2+k} > \frac{k}{\varepsilon (2k^2+3k+1)} + \frac{\varepsilon (2k^2+5k-1)}{2k^2+5k+2},
\end{equation*}
where $k \geq 1$. 
We found that the above inequality holds for 
$$0 < \varepsilon < \sqrt{\frac{k+2}{k^3+3k^2-2}}.$$
%	\begin{align*}
%	SV_{\mathcal{T}_1}^{\min} > SV_{\mathcal{T}_2}^{\max} 
%	&\Leftrightarrow \frac{k+1}{\varepsilon (2k^2+k)} + \frac{\varepsilon}{2k^2+k} > \frac{k}{\varepsilon (2k^2+3k+1)} + \frac{\varepsilon (2k^2+5k-1)}{2k^2+5k+2} \\
%	&\Leftrightarrow k > \sqrt{3} -1 \, \text{ and } \, 0 < \varepsilon < \sqrt{\frac{k+2}{k^3+3k^2-2}}.
%	\end{align*}
This means that for $0 < \varepsilon < \sqrt{\frac{k+2}{k^3+3k^2-2}}$, we have $SV_{\mathcal{T}_1}^{\min} > SV_{\mathcal{T}_2}^{\max} $, and thus all species in $\mathcal{T}_1$ have a higher Shapley value than the species in $\mathcal{T}_2$, and in particular the top $k'$ species with $k' \leq k$ are all leaves of $\mathcal{T}_1$, which completes the proof.
\qed
\end{proof}

\begin{remark}
Note that following the greedy approach of Theorem \ref{theorem4}, we would first select a pair of species $s_i^1$ and $s_j^2$, where $s_i^1$ is in $\mathcal{T}_1$ and $s_j^2$ is in $\mathcal{T}_2$ and then sequentially add $k'-2$ leaves of $\mathcal{T}_1$ or $\mathcal{T}_2$ that maximize the increase in the $PD$ score. 
Recall that the diversity of all $n=2k+1$ species was given by 
$$ PD(\{s_1^1, s_2^1, \ldots, s_k^1, s_1^2, s_2^2, \ldots, s_{k+1}^2\}) = 2 \varepsilon + \frac{1}{\varepsilon}.$$
Now let $\{ s_i^1, s_j^2, \tilde{s}_1, \tilde{s}_2, \ldots, \tilde{s}_{k'-2}\}$ be the set of the top $k'$ species obtained from the greedy algorithm. 
Then
$$ PD(\{ s_i^1, s_j^2, \tilde{s}_1, \tilde{s}_2, \ldots, \tilde{s}_{k'-2}\}) = a \cdot \varepsilon + \frac{1}{\varepsilon} \overset{\varepsilon \rightarrow 0}{\longrightarrow} \, \infty,$$
where $a < 2$.

This implies that the absolute difference between the diversity of the top $k' \leq k, k' \geq 2$ greedy species and the diversity of all $n=2k+1$ species may be arbitrarily small, because
\begin{align*}
\left\vert PD(\{s_1^1, s_2^1, \ldots, s_k^1, s_1^2, s_2^2, \ldots, s_{k+1}^2\}) - PD(\{ s_i^1, s_j^2, \tilde{s}_1, \tilde{s}_2, \ldots, \tilde{s}_{k'-2}\}) \right\vert &= \left\vert 2 \varepsilon + \frac{1}{\varepsilon} - (a \cdot \varepsilon + \frac{1}{\varepsilon}) \right\vert \\
&= \left\vert (2-a) \varepsilon \right\vert \overset{\varepsilon \rightarrow 0}{\longrightarrow} \, 0.
\end{align*}
On the other hand, the absolute difference between the diversity of the top $k'$ greedy species and the diversity of the top $k'$ Shapley species with $k' \leq k, k' \geq 2$ may be arbitrarily large, because
\begin{align*}
\left\vert \underbrace{PD(\{ s_i^1, s_j^2, \tilde{s}_1, \tilde{s}_2, \ldots, \tilde{s}_{k'-2}\}}_{\geq \frac{1}{\varepsilon}} - \underbrace{PD(\{s_1, \ldots, s_k'\}}_{\leq \varepsilon} \right\vert 
\geq \left\vert \frac{1}{\varepsilon} - \varepsilon \right\vert \overset{\varepsilon \rightarrow 0}{\longrightarrow} \, \infty.
\end{align*}
\end{remark}

\section{Conclusions}
In this paper we have considered the Shapley value of unrooted phylogenetic trees and have answered some of the questions posed in \citet{Haake2008}. 
Firstly, we have revisited the linear transformation that allows for the calculation of the Shapley value from the edge lengths of a trees (cf. \citet{Haake2008}) and have shown that non-isomorphic trees may have permutation-equivalent Shapley transformation matrices and permutation-equivalent null spaces. This implies that neither estimating or determining the so-called split counts associated with a tree nor the Shapley values of its leaves suffice to reconstruct the correct tree topology. 
Note that deciding whether two non-isomorphic trees have permutation-equivalent Shapley transformation matrices may be a hard problem, because it can be related to the so-called \emph{graph isomorphism problem}, whose complexity is not known. Given two finite graphs $G_1=(V_1, E_1)$ and $G_2=(V_2,E_2)$ with $\vert V_1 \vert = \vert V_2 \vert$ and $\vert E_1 \vert = \vert E_2 \vert$, the graph isomorphism problem asks whether $G_1$ and $G_2$ are isomorphic. Let $\vert V_1 \vert = \vert V_2 \vert = n$ and $\vert E_1 \vert = \vert E_2 \vert = m$ and let $I_1$ be the $n \times m$ incidence matrix of $G_1$ and let $I_2$ be the $n \times m$ incidence matrix of $G_2$ (i.e., $I_1[i,j] = 1$ if vertex $i \in V_1$ and edge $j \in E_1$ are incident in $G_1$ and $0$ otherwise (analogously for $I_2$)). Then $G_1$ and $G_2$ are isomorphic if and only if their incidence matrices are permutation-equivalent, i.e., if there exists a permutation matrix $P \in \mathbb{R}^{n \times n}$ and a permutation matrix $Q \in \mathbb{R}^{m \times m}$ such that $P \, I_1 \, Q = I_2$.
Even though Shapley transformation matrices are not incidence matrices (because their entries are different from $0$ and $1$), the problem of deciding whether they are permutation-equivalent or not may be related to the problem of deciding whether two incidence matrices are permutation-equivalent and thus, the problem may be related to the graph isomorphism problem. A direction for future research could therefore be to further analyze the relationship between the graph isomorphism problem and the question whether two Shapley transformation matrices are permutation-equivalent. It would also be of interest to assess the complexity of deciding whether two Shapley transformation matrices are permutation-equivalent or not.

Another direction of further research could be the use of the Shapley value as a conservation criterion in preservation. 
In this manuscript we have considered the application of the Shapley value as a prioritization criterion in a simple variant of the Noah's ark problem (\citet{Weitzman1998}) and compared it to a greedy algorithm (\citet{Steel2005}). It turned out that the Shapley value may perform very badly as a prioritization criterion, meaning that the diversity of the top $k$ species (ranked by their Shapley values) may not approximate the total diversity of all $n$ species at all. Thus, in this case using the Shapley value in order to find a subset of species of size $k$ that maximizes the $PD$ score cannot compete with the greedy algorithm introduced in \citet{Steel2005}. 
Note, however, that our class of trees where the diversity of the top $k' \leq k$ species tends to zero while the diversity of all $2k+1$ species tends to infinity, only works if $k'$ is less than half of the number of all species. It would be of interest to see, whether a similar construction can be found in order to show that the Shapley value will fail as a prioritization criterion for any $k' \leq n$.
It would also be of interest to see if a better performance of the Shapley value can be guaranteed when turning from unrooted to rooted phylogenetic trees. \\

%\begin{figure}
%  \includegraphics{example.eps}
%\end{figure}
%\begin{figure*}
%  \includegraphics[width=0.75\textwidth]{example.eps}
%\end{figure*}
%\begin{table}
%\begin{tabular}{lll}
%\hline\noalign{\smallskip}
%\end{tabular}
%\end{table}

\begin{acknowledgements}
The first author thanks the Ernst-Moritz-Arndt-University Greifswald for the Landesgraduiertenförderung studentship, under which this work was conducted, and the Barcelona Graduate school of Mathematics (BGSMath) for financial support for attending the Algebraic and Combinatorial Phylogenetics program in Barcelona in June 2017, during which some of the results presented in this manuscript were obtained.
\end{acknowledgements}

\newpage
\bibliographystyle{spbasic}      % basic style, author-year citations
\bibliography{Literatur}   % name your BibTeX data base

%\begin{thebibliography}{}
% etc
%\end{thebibliography}

\begin{appendix}
\section{Appendix} 
In order to find non-isomorphic trees with permutation-equivalent Shapley transformation matrices we have exhaustively analyzed all tree topologies up to 17 taxa and their split counts. To be precise, we have considered different necessary (but not sufficient) conditions for two non-isomorphic trees to have permutation-equivalent Shapley transformation matrices, the details of which will be explained in the following. Note that we have considered these necessary (but not sufficient) conditions as a first step, because they can quickly be checked, while directly examining whether two matrices a permutation-equivalent is time-consuming and not feasible for large matrices. Using these necessary conditions we have performed a candidate search for trees inducing permutation-equivalent Shapley transformation matrices, where the candidates were then further analyzed. We now describe the necessary conditions we used.

\begin{enumerate}

\item Split size sequence: \\
	Recall that the Shapley transformation matrix of a tree $\mathcal{T}$ solely depends on the splits counts associated with its edges (cf. Theorem \ref{theorem1}).
	In order for two tree topologies to have permutation-equivalent Shapley transformation matrices, they must exhibit the same split counts, in particular they must exhibit the same \emph{split sizes}, where for a split $\sigma = A \vert B$ with $A, B \subset X, \, A \cap B = \emptyset,$ and $A \cup B$ we let $\| \sigma \| = \| A \vert B \| = \min \{\vert A \vert, \vert B \vert \}$ denote its size. Note that any binary tree $\mathcal{T}$ on $n$ leaves induces $n$ trivial splits (where either $\vert A \vert = 1$ or $\vert B \vert  = 1$) and $n-3$ non trivial splits. Following \citet{Fischer2015a} we assume an arbitrary ordering of these splits $\sigma_1, \ldots, \sigma_{n-3}$ and define the $(n-3)$ tuple $\tilde{s}(\mathcal{T})$ as follows:
	$$ \tilde{s}(\mathcal{T})_i = \| \sigma_i \| \text{ for all } i=1, \ldots, n-3.$$
We now order the $n-3$ entries of $\tilde{s}(\mathcal{T})$ increasingly and call the resulting ordered sequence the \emph{split size sequence} $s(\mathcal{T})$. Now, in order for two trees to have permutation-equivalent Shapley transformation matrices, their split size sequences must be identical, which gives us a first necessary condition. For $\mathcal{T}_1$ and $\mathcal{T}_2$ depicted in Figure \ref{Fig3} we for example have \\ $s(\mathcal{T}_1) = s(\mathcal{T}_2) = (2,2,2,2,2,2,3,3,4,4,4,4,8,8)$.
\item Matrix entries: \\
If two trees exhibit the same split size sequence, we compute their Shapley transformation matrices and analyze them:
	\begin{enumerate} 
	\item For two matrices $\mathbf{M}_1$ and $\mathbf{M}_2$ to be permutation-equivalent, they must contain the same entries. To check if this is the case, we \enquote{flatten} both matrices and define $s(\mathbf{M}_1)$ to be the sequence containing all matrix elements of $\mathbf{M}_1$ in an increasing order and analogously we define $s(\mathbf{M}_2)$ to be the sequence containing all entries of $\mathbf{M}_2$ ordered increasingly. If $s(\mathbf{M}_1) = s(\mathbf{M}_2)$, the two matrices share the same entries and we proceed with a subsequent analysis of rows and columns.
	\item Recall that two matrices are permutation-equivalent if they are identical up to a permutation of rows and columns. Thus, we derive two additional necessary conditions for two matrices to be permutation-equivalent.
		 \begin{itemize}
		 \item For all rows $r_i^1$ of $\mathbf{M}_1$ we define $s(r_i^1)$ to be the sequence containing the elements of $r_i^1$ in an increasing order. Analogously we define $s(r_j^2)$ to be the sequence containing the elements of a row $r_j^2$ of matrix $\mathbf{M}_2$. Now for all rows $r_i^1$ of $\mathbf{M}$ we check if $s(r_i^1) = s(r_j^2)$ for some row $r_j^2$ of $\mathbf{M}_2$.
		 \item Similarly, we compare the columns of $\mathbf{M}_1$ and $\mathbf{M}_2$. For any column $c_i^1$ of $\mathbf{M}_1$ or $c_j^2$ of $\mathbf{M}_2$ we define $s(c_i^1)$ and $s(c_j^2)$ to be the sequence containing the elements of the corresponding column in an increasing order. Now for all columns $c_i^1$ of $\mathbf{M}$ we check if $s(c_i^1) = s(c_j^2)$ for some column $c_j^2$ of $\mathbf{M}_2$.
		 \end{itemize}
	\end{enumerate}
\end{enumerate}
We now summarize the above conditions in the following algorithm (Algorithm \ref{Algorithm}) that checks whether two non-isomorphic trees $\mathcal{T}_1$ and $\mathcal{T}_2$ are candidates for trees inducing permutation-equivalent Shapley transformation matrices.
%%%%
\begin{algorithm}

\KwIn{two non-isomorphic trees $\mathcal{T}_1$ and $\mathcal{T}_2$ on $n$ leaves} 
\KwOut{TRUE (trees are possible candidates and have to be further analyzed) or FALSE (trees do not have permutation-equivalent Shapley transformation matrices).}
Compute the split size sequences $s(\mathcal{T}_1)$ and $s(\mathcal{T}_2)$\;
\eIf{$s(\mathcal{T}_1) \neq s(\mathcal{T}_2)$}{
	\Return FALSE\;}
{
	Compute the Shapley transformation matrices $\mathbf{M}_1$ and $\mathbf{M}_2$ and flatten them to the sequences $s(\mathbf{M}_1)$ and $s(\mathbf{M}_2)$ (containing all matrix elements in an increasing order)\;
	\eIf{$s(\mathbf{M}_1) \neq s(\mathbf{M}_2)$} {
		\Return FALSE }
	{
		\ForAll{rows $r_i^1$ of $\mathbf{M}_1$}{
			Sort the entries of $r_i^1$ increasingly and compare this sorted sequence $s(r_i^1)$ to all sorted rows of $\mathbf{M}_2$\;
			\eIf{$s(r_i^1) \neq s(r_j^2)$ for all $j=1, \ldots, $ number of rows of $\mathbf{M}_2$}{
				\Return FALSE;}
			{
			\ForAll{columns $c_i^1$ of $\mathbf{M}_1$}{
				 Sort the entries of $c_i^1$ increasingly and compare this sorted sequence $s(c_i^1)$ to all sorted columns of $\mathbf{M}_2$\;
					\eIf{$s(c_i^1) \neq s(c_j^2)$ for all $j=1, \ldots, $ number of columns of $\mathbf{M}_2$}{
						\Return FALSE\;}
					{
					\Return TRUE\;
					}
				}
			}
		}
	}
}
%\end{algorithmic}
\caption{Permutation-equivalent Shapley transformation matrices -- Candidate Search}\label{Algorithm}
\end{algorithm}
%%%
\hfill \\Note that the algorithm returns TRUE, if the input trees possibly induce permutation-equivalent Shapley transformation matrices and FALSE if this can be ruled out (i.e., any of the necessary conditions introduced above is violated). However, if the algorithm returns TRUE the possible candidates have to be further analyzed, as all conditions mentioned above are necessary for two trees to have permutation-equivalent Shapley transformation matrices, but not sufficient (cf. Example \ref{ex_not_sufficient}). However, we have conducted this candidate search in Mathematica \citet{Mathematica} and have analyzed all tree topologies up to 16 leaves. The only pair of candidates that we found is the pair $(\mathcal{T}_1, \mathcal{T}_2)$ depicted in Figure \ref{Fig3} and used in the proof of Theorem \ref{observation1}. Thus, this pair is the smallest example for a pair of non-isomorphic trees inducing permutation-equivalent Shapley matrices (and thus permutation-equivalent null spaces). Subsequently, we have looked at the case of 17 taxa, where again only one pair of candidate trees was found (trees $\mathcal{T}_1'$ and  $\mathcal{T}_2'$ depicted in Figure \ref{Fig6}). However, as we will explain below, $\mathcal{T}_1'$ and $\mathcal{T}_2'$ do \emph{not} induce permutation-equivalent Shapley transformation matrices, which illustrates the fact that the conditions described above and used in Algorithm \ref{Algorithm} are only necessary, but not sufficient conditions.

\begin{example} \label{ex_not_sufficient}
Consider the pair of trees $(\mathcal{T}_1', \mathcal{T}_2')$ on 17 leaves depicted in Figure \ref{Fig6}. Algorithm \ref{Algorithm} returns TRUE for this pair of trees, i.e., $\mathcal{T}_1'$ and $\mathcal{T}_2'$ are possible candidates for two non-isomorphic trees inducing permutation-equivalent Shapley transformation matrices. However, their Shapley transformation matrices are not permutation-equivalent. To see this, consider the split counts associated with edge $I_{13}$ and compare them for $\mathcal{T}_1'$ and $\mathcal{T}_2'$ (cf. Table \ref{Table_SplitCounts}).
	\begin{table} 
	\centering
	\caption{Split counts induced by edge $I_{13}$}
	\begin{tabular}{c|cc|cc}
	\toprule
	Leaf $i$ & $f_{\mathcal{T}_1}(i, I_{13})$ & $n_{\mathcal{T}_1}(i, I_{13})$ & $f_{\mathcal{T}_2}(i, I_{13})$ & $n_{\mathcal{T}_2}(i, I_{13})$ \\
	\midrule
	1 & 9 & 8 & 9 & 8 \\
	2 & 9 & 8 & 9 & 8 \\
	3 & 9 & 8 & 9 & 8 \\
	4 & 9 & 8 & 9 & 8 \\
	\midrule
	\textbf{5} & 9 & 8 & 8 & 9 \\
	\textbf{6} & 9 & 8 & 8 & 9 \\
	\textbf{7} & 9 & 8 & 8 & 9 \\
	\textbf{8} & 9 & 8 & 8 & 9 \\
	\midrule
	9 & 8 & 9 & 8 & 9 \\
	10 & 8 & 9 & 8 & 9 \\
	11 & 8 & 9 & 8 & 9 \\
	12 & 8 & 9 & 8 & 9 \\
	\midrule
	\textbf{13} & 8 & 9 & 9 & 8 \\
	\textbf{14} & 8 & 9 & 9 & 8 \\
	\textbf{15} & 8 & 9 & 9 & 8 \\
	\textbf{16} & 8 & 9 & 9 & 8 \\
	\midrule
	17 & 8 & 9 & 8 & 9 \\
	\bottomrule
	\end{tabular} 
	\label{Table_SplitCounts}
	\end{table}
For leaves $1,2,3,4,9,10,11,12$ and $17$ edge $I_{13}$ induces the same split counts in both $\mathcal{T}_1'$ and $\mathcal{T}_2'$. However, for leaves $5, 6,7, 8$ and leaves $13,14,15, 16$ the split counts differ. To be precise, we have $f_{\mathcal{T}_1'}(i,I_{13}) = 9$ and $f_{\mathcal{T}_2'}(i,I_{13}) = 8$ for $i = 5,6,7,8$ and $f_{\mathcal{T}_1'}(j,I_{13}) = 8$ and $f_{\mathcal{T}_2'}(j,I_{13}) = 9$ for $j = 13,14,15,16$. At first glance we can make the split counts associated with edge $I_{13}$ coincide for $\mathcal{T}_1'$ and $\mathcal{T}_2'$ by swapping leaves $5,6,7,8$ with leaves $13,14,15,16$ in $\mathcal{T}_2'$ (i.e., by permuting the rows associated with these leaves in the Shapley transformation matrix). However, then the split counts induced by for example edge $I_5$ will differ between $\mathcal{T}_1'$ and $\mathcal{T}_2'$. It can be checked that no permutation of rows or columns of the Shapley transformation matrix $\mathbf{M}_2'$ of $\mathcal{T}_2'$ exists such that it coincides with the Shapley transformation matrix $\mathbf{M}_1'$ of $\mathcal{T}_1'$. Thus, the Shapley transformation matrices of $\mathcal{T}_1'$ and $\mathcal{T}_2'$ are not permutation-equivalent even though Algorithm \ref{Algorithm} suggests them as candidates. This shows that the criteria used in Algorithm \ref{Algorithm} are necessary but not sufficient conditions for two non-isomorphic trees to have permutation-equivalent Shapley transformation matrices.
\end{example}

\begin{figure} [htbp]
	\centering
	\includegraphics[scale=0.9]{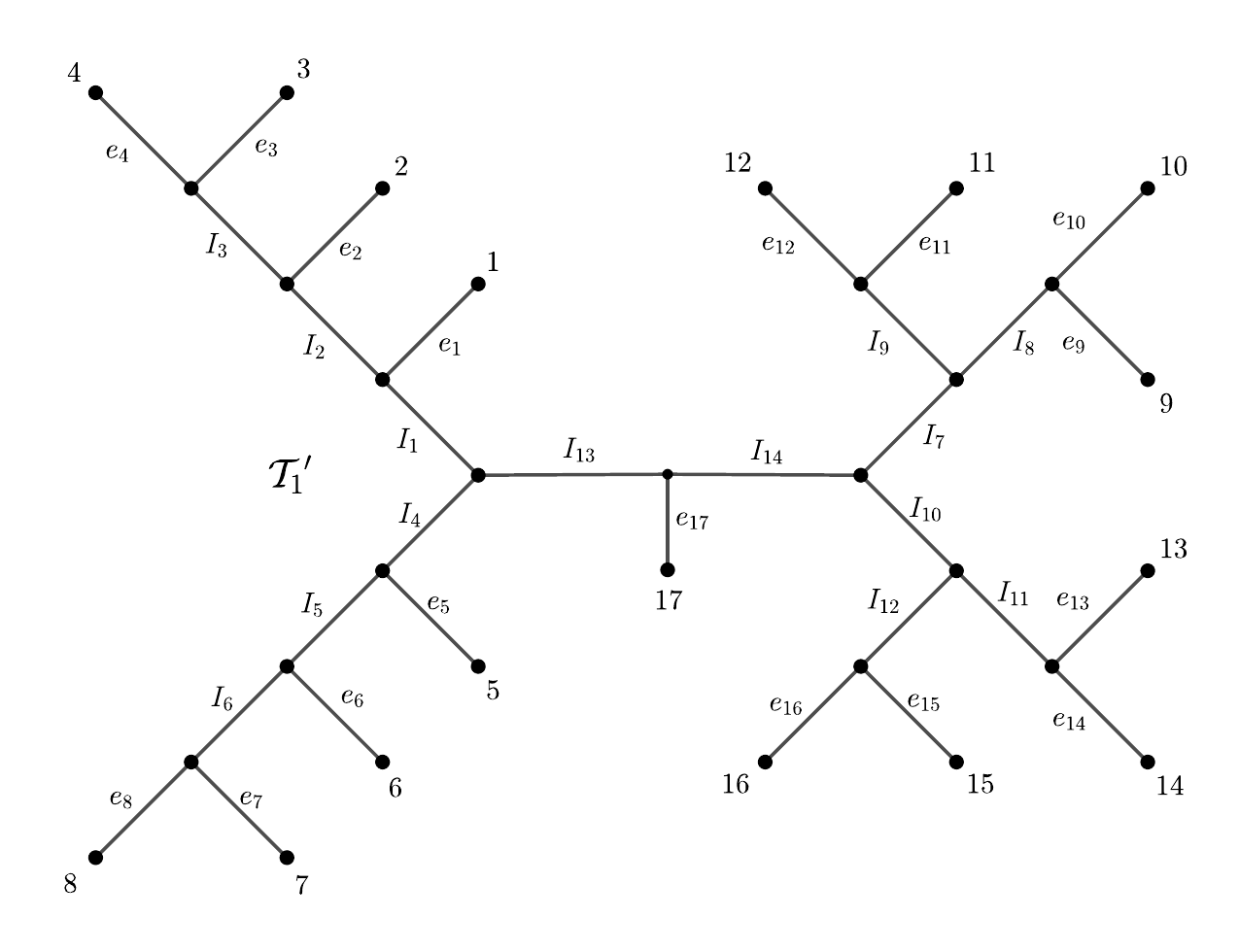} \\
	\includegraphics[scale=0.9]{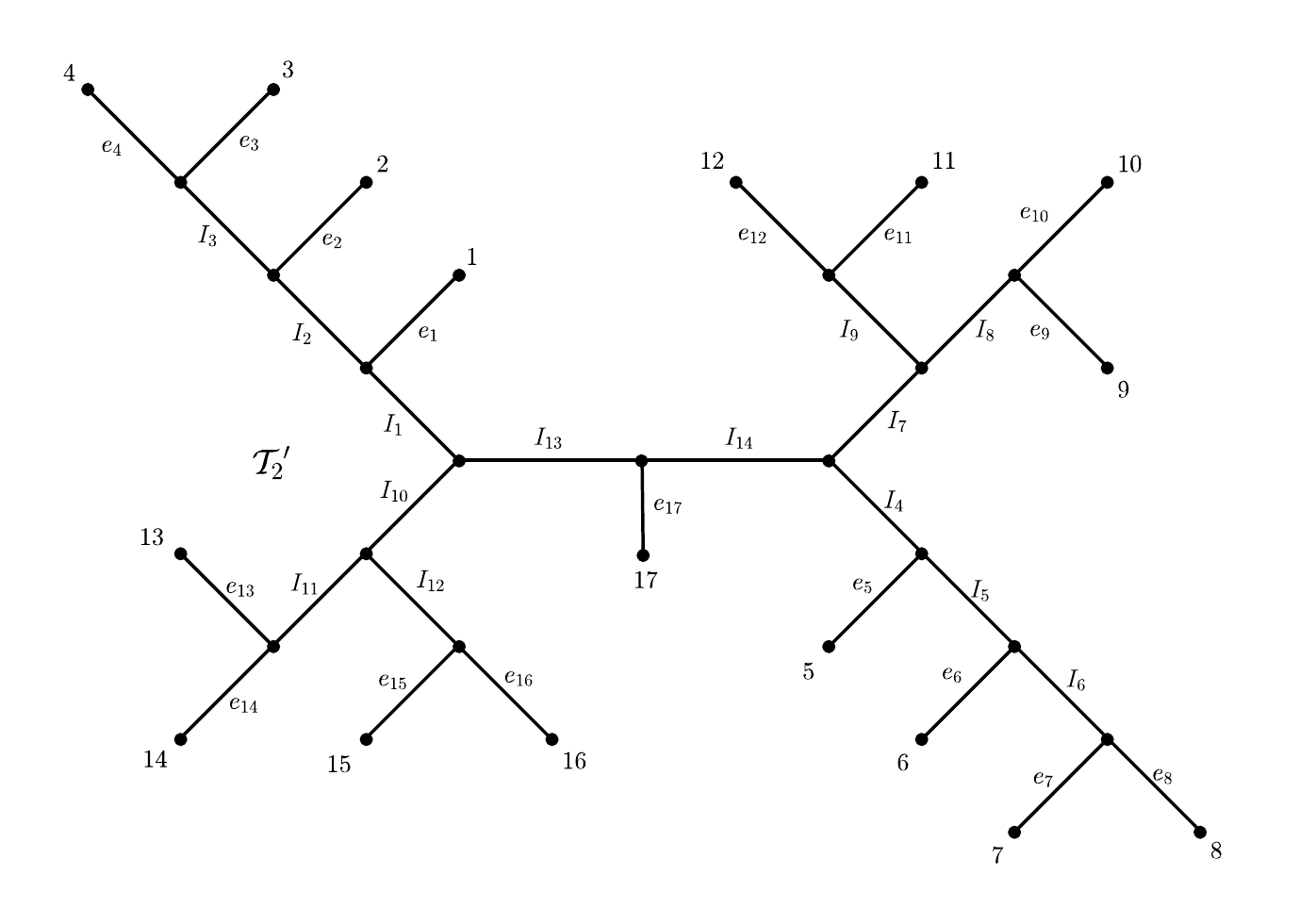}
	\caption{Two non-isomorphic trees $\mathcal{T}_1'$ and $\mathcal{T}_2'$ on 17 leaves that are found by Algorithm \ref{Algorithm} but do not induce permutation-equivalent Shapley transformation matrices and permutation-equivalent null spaces.}
	\label{Fig6}
\end{figure}

\end{appendix}

\end{document}